\documentclass[11pt, a4paper]{article}
\usepackage[a4paper, top=2.5cm, bottom=2.5cm, left=2cm, right=2cm]{geometry}

\usepackage{amsmath, amssymb, amsthm}
\usepackage{hyperref}
\usepackage{cite}
\usepackage{algorithm}
\usepackage{algpseudocode}
\usepackage{mathtools}
\usepackage{bm}
 
\algnewcommand{\LineFor}[2]{\State \textbf{for} #1 \textbf{do} #2 \textbf{end for}}

\usepackage{xcolor}
\usepackage{comment}
\usepackage{bm}

\title{Streaming with Catalytic Memory\footnote{An extended abstract of this paper appears in ESA 2026.}}
\author{Tamara Kaplan, Nimrod Kaplan, Haim Kaplan}

\theoremstyle{plain}

\newtheorem{theorem}{Theorem}[section]
\newtheorem{lemma}[theorem]{Lemma}

\newtheorem{claim}[theorem]{Claim}
\newtheorem{assumption}[theorem]{Assumption}
\newtheorem{proposition}[theorem]{Proposition}

\newtheorem{observation}[theorem]{Observation}

\theoremstyle{definition}
\newtheorem{definition}[theorem]{Definition}

\theoremstyle{remark}
\newtheorem{remark}[theorem]{Remark}

\newcommand{\F}{\mathbb{F}}
\newcommand{\x}{\boldsymbol{x}}
\newcommand{\y}{\boldsymbol{y}}
\newcommand{\0}{\boldsymbol{0}}

\newcommand{\btau}{\bm{\mathcal{T}}}
\newcommand{\bsig}{\boldsymbol{\sigma}}
\newcommand{\disj}{\text{Disj}}
\newcommand{\N}{\mathbb{N}}

\newcommand{\defsymbol}{\hfill $\lozenge$}

\newcommand{\Input}{\Require}
\newcommand{\Output}{\Ensure}

\begin{document}
\maketitle

\begin{abstract}
We introduce a streaming model that uses both catalytic and regular memory. In this model, we show how to exactly compute the frequency moments using a logarithmic number of bits of regular memory and a polynomial number of bits of catalytic memory. More generally, we show how to compute arbitrary polynomials of the item frequencies exactly within the same space bounds. As an application, we obtain catalytic streaming algorithms that exactly compute the number of distinct elements in a stream, count the number of triangles (or any other small subgraph) in a graph whose edges arrive in a stream, and identify heavy hitters.

Our algorithms for frequency moments perform a constant number of passes over the stream, and for polynomial evaluation, we require one more pass than the degree of the polynomial. 
By relating our catalytic streaming model to the catalytic communication model introduced in \cite{catalyticCommunication}, we show that catalytic memory is not useful for any one-pass streaming algorithm. For lower bounds on multipass streaming algorithms, the impossibility results of \cite{catalyticCommunication}  are not strong enough.
However, using a different technique, we show that under certain natural restrictions, no catalytic streaming algorithm can compute the second frequency moment in fewer than three passes.

This definition of the restricted class of  two-pass algorithms then guides us in the design of a two-pass algorithm 
for computing the second moment exactly
that circumvents these restrictions  and breaks the three-pass barrier. 

\end{abstract}

\section{Introduction}
Catalytic computation was introduced by Buhrman, Cleve, Kouck\'y, Loff, and Speelman~\cite{buhrman2014computing} as a theoretical model for understanding the computational power of a ``full hard drive'' (i.e.,\ memory which can be used but has to be restored at the end to its original content). Since then, catalytic memory has attracted significant attention and has led to several interesting results in complexity theory. In particular, researchers have explored the power of non-determinism~\cite{buhrman2018catalytic} and randomness~\cite{rand-vs-comp25} in this model, as well as non-uniform~\cite{girard2015nonuniform, potechin2017note, cook2022trading} and quantum~\cite{buhrman2025quantum} versions of it. Furthermore, interesting space-efficient algorithms using catalytic memory have been developed. 
 Notable examples include the work of Henzinger, Pyne, and Ragavan \cite{henzinger2026catalytic}, who improved upon Cook and Mertz's \cite{CMtree24} results by solving the Tree Evaluation Problem using subpolynomial catalytic memory and logarithmic regular memory. Additionally, Chmel et al. \cite{chmel2026} provided deterministic solutions for directed connectivity and sequence alignment.

A \emph{Catalytic}
Turing machine has a working tape as of a standard Turing machine and in addition it has a catalytic tape initialized arbitrarily that has to be restored to its original content by the end of the computation.\footnote{The terminology comes from chemistry: a catalyst enables a reaction to proceed without being consumed or permanently altered.}
The class \emph{Catalytic Logspace}, denoted by $\mathsf{CL}$, 
 contains all problems computable by a catalytic Turing machine with a regular tape of size logarithmic in the input length and a polynomial-size catalytic tape. Surprisingly, Buhrman et al.~\cite{buhrman2014computing} proved that 
\(
\mathsf{TC}^1 \subseteq \mathsf{CL},
\)\footnote{$\mathsf{TC}^1$ denotes the class of decision problems solvable by a family of polynomial-size, logarithmic-depth Boolean circuits with unbounded fan-in AND, OR, and MAJORITY (aka threshold) gates.} 
showing that catalytic memory can be exploited to solve problems not known to lie in $\mathsf{L}$.\footnote{$L$ is the class of problems computable  with memory of size logarithmic in the input length.}

Recently, Pyne, Sheffield, and Wang~\cite{catalyticCommunication} introduced a \emph{catalytic communication} model,
in which Alice and Bob communicate by exchanging messages using a small clean memory and a large catalytic memory. They show that this model is substantially stronger than the standard one as it enables them to solve certain problems such as set equality or inner product over $GF(2)$, exactly, using only one bit of clean memory.

Streaming algorithms typically have memory which is only logarithmic in the input size (so they cannot store their inputs). They traverse their inputs one by one in a single or multiple passes, maintain some information in their small memory and compute the result. Streaming algorithms have been intensively studied for many problems, starting from the seminal work of Alon, Matias, and Szegedy on the frequency moment \cite{AMS}, for a survey see \cite{Muthukrishnan2005,McGregor2014,ChakrabartiStreamingNotes}.

\subsection{The Catalytic Streaming model}
\label{sec:model}

In the standard multi-pass streaming model \cite{munro1980selection} (see also the surveys cited above)
a  sequence of $m$ elements $\bsig =x_1, x_2, \ldots, x_m$ is given sequentially as input. The algorithm is then allowed to make $p$  passes over $\sigma$ while  maintaining $s$ (sublinear in $m$) bits of memory.

In the \emph{catalytic streaming model} in addition to the $s$ bits of regular memory, the algorithm has access to a large catalytic memory of $c$ bits. Initially, the catalytic memory contains an arbitrary string $\bm{\mathcal{T}}$ of $c$ bits. The algorithm can use this catalytic memory however it wants, with the constraint that it must be restored to contain exactly $\bm{\mathcal{T}}$ when the algorithm ends.
Its formal definition is as follows.

\begin{definition}[Catalytic Streaming] \label{def:cat streaming}

\looseness = -1
A \emph{catalytic streaming algorithm} $A$ has a read-only input tape, a read-only stream tape (storing the current stream element), a read-write \emph{regular memory}, and a read-write \emph{catalytic memory}. 
$A$ gets $n,m\in \N$ as input, possibly additional inputs on its input tape, and a stream
$\bsig = x_1,x_2,\ldots,x_m$ of elements from a domain $D$ of size $n$.\footnote{
We assume in this paper that $n\leq m$. 
Since the number of catalytic registers which we use depends on $n$, then
when $n>m$, we need $\log n$ (rather than $\log m)$ regular bits to address the catalytic registers.}

Let $p \in \N$. We say that $A$ makes $p$ \emph{passes} if it scans its input stream $\sigma$ sequentially from $x_1$ to $x_m$, $p$ times. Thus, during each pass, at time step $j$, the symbol $x_j$ appears on the stream tape and is processed by $A$.

Suppose that initially the catalytic memory contains an arbitrary string $\bm{\mathcal{T}}$. We say that  $A$ computes the function $F$ from $D^m$ to some output range $\mathcal{O}$ in $p$ passes if, for every stream $\bsig = x_1,x_2,\ldots,x_m$, after the $p$-th pass of $A$ on input $\bsig$:
\begin{enumerate}
    \item The value $F(x_1,x_2,\dots,x_m)$ appears in memory, and
    \item The catalytic memory  is restored to contain exactly the string $\bm{\mathcal{T}}$. 
\end{enumerate}

We say that $A$ uses \emph{$s$ bits of regular memory} and \emph{$c$ bits of catalytic memory} if, at every point during its execution, it uses at most $s$ bits on the regular memory and at most $c$ bits on the catalytic memory. \defsymbol
\end{definition}

\begin{remark}
Unless stated otherwise, the stream elements in this paper are drawn from $[n]\coloneqq\{1,2,\ldots,n\}$.  
Although our complexity measure is in bits, we describe our algorithms in terms of \emph{registers} storing integers.
We perform our computations over a ring of integers modulo some $q \in \N$, which we denote by $Z_q$. Whenever additional assumptions on the ring are needed, we will state them explicitly.
\end{remark}

Basic  catalytic algorithms rely on a symmetric execution—computing forward, copying the output, and uncomputing backward to restore the memory. This often necessitates reading the input in reverse order during the uncomputation phase. In the streaming model, however, data arrives sequentially and can only be accessed in the forward direction. This restriction breaks the standard uncomputation strategy, making the design of catalytic streaming algorithms a significant challenge even with multiple passes.

\subsection{Our contributions}
We introduce the
catalytic streaming model (\autoref{def:cat streaming}), describe algorithms for classical streaming problems in this model, and  prove  lower bounds as follows.

\paragraph*{Frequency Moments}
We show that a straightforward generalization of the catalytic exponentiation algorithm given in \cite{buhrman2014computing} can be used to compute $F_k(\bsig)=\sum_{i\in [n]}f_k^k$ exactly,  using four passes over $\sigma$ for any $k\ge 2$.
Here $
f_i \coloneqq \bigl|\{j \in [m] : x_j = i\}\bigr| $ is the number of occurrences of item $i$ in $\sigma$.
Computing frequency moments is maybe the most classical and well-studied streaming problem. This result sets a clear separation between the model of catalytic streaming and the standard streaming model, since it's  well-known that any moment other than $F_1$ cannot be computed exactly in sublinear space using a constant number of passes.
Classical results give randomized sublinear algorithms for approximating the frequency moments \cite{AMS}. 

Using a different technique, we also show how to compute the second frequency moment exactly in two passes, improving upon the four-pass algorithm for general frequencies.

\paragraph*{Polynomial Evaluation}
We give a new, simple, catalytic streaming algorithm that evaluates $P(f_1,\ldots,f_n)$ in $k+1$ passes where $P$ is a multivariate polynomial of 
degree $k$. 
We can evaluate this polynomial over any ring $Z_q$ (the integers modulo $q$) as long as this ring contains multiplicative inverses to $2,3,\ldots,k$. We can do so by using $O(1)$ regular registers and $n$ catalytic registers of
$O(\log q)$ bits each.
Since the evaluation of $P(f_1,f_2,\ldots,f_n)$ over the integers is $O(m^k)$,\footnote{The hidden constant depends on the largest coefficient of a monomial.}, by choosing $q$ of this magnitude, we in fact compute $P(f_1,\ldots,f_n)$ over the integers.

\paragraph*{Lower Bounds}
We were intrigued by the question of whether we can compute $F_2$ in the catalytic streaming model in less than three passes.
Using results in the recently introduced 
catalytic communication complexity model~\cite{catalyticCommunication},
and a well-known reduction
from the \textsc{Set-Disjointness} problem in communication complexity to streaming algorithms for frequency moments,  we are able to easily show that one cannot compute $F_k$ for any $k\ge 2$ using a one-pass catalytic streaming algorithm and only sublinear regular memory.
However, in the catalytic arena this approach fails for two passes since, as we show, there is a catalytic communication protocol with small regular memory for \textsc{Set-Disjointness} in three rounds. 

This led us to characterize a class of two-pass catalytic streaming algorithms for which computing the second frequency moment using sublinear regular memory is impossible. 
The understanding of the  limitations of this class enabled us to develop a two-pass algorithm for $F_2$ that avoids them and overcomes the three-pass barrier.
We believe that our restricted family is still of interest as 
a natural extension of this family to algorithms with more than two passes includes all algorithms presented here, the only exception is the 2-pass algorithm we mentioned above.

\paragraph*{Applications: $F_0$, Counting Subgraphs, and Heavy Hitters}
Finally, we apply our algorithms for moments and polynomial evaluation to obtain space efficient catalytic streaming algorithms for several other well-studied streaming problems. In particular, we get a four-pass algorithm to compute the exact number of distinct items in a stream. We show how to count triangles (or any other small subgraph)
in a multigraph whose edges arrive as a stream. We also show that we can compute the exact set of $F_k$ heavy hitters in $O(\log n)$ passes.\footnote{All items $i$ such that $f_i^k\ge \epsilon F_k$.} 

\subsection{Roadmap}
We  follow the notation of \autoref{def:cat streaming}.
\autoref{sec:moments} shows how to compute moments in four passes. \autoref{sec:poly} gives our polynomial evaluation algorithm. We prove lower bounds in \autoref{sec:lower bounds}. \autoref{sec:app} shows how to apply our algorithm to compute the number of distinct elements in a stream (so called $F_0$), count small subgraphs in a graph stream, and compute heavy hitters.

\section{Computing Exact Moments Using 4 Passes and Logarithmic Space} \label{sec:moments}

In this section, we study the classic problem of computing the $k$-th frequency moment of a stream and prove the following theorem.

\begin{theorem} \label{thm:calculating F_k}
    Let $\bsig=x_1,x_2,\dots,x_m$ be a stream over $[n]$, and let $f_1,f_2,\dots,f_n$ denote the corresponding frequencies.
    Then there exists a catalytic streaming algorithm that evaluates
    \(
    F_k(\bsig) \coloneqq \sum_{i=1}^n f_i^k
    \)
    in $4$ passes, using $O(1)$ regular registers and $n \cdot (k+1)$ catalytic registers. This algorithm works over any ring $Z_q$. To compute $F_k$ over the integers it is enough to take $q = O(m^k)$.
\end{theorem}

A key lemma in the proof of the above theorem is a result proved in \cite{buhrman2014computing} called \emph{the powering lemma}. We formulated it here in a suitable way for our task.

\begin{lemma}[{\cite[Powering, Lemma~10 (modified)]{buhrman2014computing}}]\label{Powering Lemma}
Let $k \in \N$. Let $r$ and $r_1,r_2,\ldots,r_k$  be catalytic registers, initialized with arbitrary values $\tau$, $\tau_1, \tau_2, \ldots, \tau_k$ respectively, and let $r_o$ be a regular register. There are programs $I_1$, $I_2$, and $I_3$ that only use the registers $\{r_i\}_{i \in [k]}$ and $r_o$ such that for every input $f \in [n]$ the program
\[
I_1,\;  r\leftarrow r+f,\; I_2,\; r\leftarrow r-f,\; I_3
\]
computes
\(
r_o \leftarrow f^k.
\)
Moreover, the values of the other registers are as follows
\[
r=\tau, \quad \text{and for every $i\in [k]$,} \quad r_i=\tau_i-(\tau+f)^{k-i}.
\]  They can be restored by executing the program in reverse.
\end{lemma}

For convenience, we give in \autoref{app:powering} the algorithm used in \cite{buhrman2014computing} to prove \autoref{Powering Lemma}.

\begin{proof}[Proof of \autoref{thm:calculating F_k}]
Let $r_1,r_2,\ldots,r_n$ and $\{r_{i,j}\}_{i \in [n],\, j \in [k]}$ be catalytic registers, and let $r_o$ be a regular output register. Given \autoref{Powering Lemma}, computing $F_k(\bsig)$ is straightforward. We simply run the algorithm of \autoref{Powering Lemma} in parallel for each item $i \in [n]$, using the frequency $f_i$ in place of the value $f$ appearing in that lemma.

All of these parallel executions share the same output register $r_o$, so their contributions accumulate to
\(
\sum_{i=1}^n f_i^k = F_k(\bsig).
\)
On the other hand, each execution uses its own catalytic registers, i.e. for the computation corresponding to $f_i$, we use $r_i$ in the role of the register $r$, and $r_{i,1},r_{i,2},\ldots,r_{i,k}$ in the roles of the registers $r_1,r_2,\ldots,r_k$. Denote by $\tau_i$ the initial value of $r_i$ and $\tau_{i,j}$ the initial value of $r_{i,j}$.

The resulting algorithm is given in \autoref{alg:calculating moments}. There, for $t\in \{1,2,3\}$, $I_t^i$ denotes the program $I_t$ applied to the registers dedicated to $f_i$. Observe that before the two additional passes mentioned in line~\ref{line:f_k comp two more pass line} the values of the registers $r_1,r_2,\ldots,r_n$ are restored, but each register $r_{i,j}$ still contains $\tau_{i,j} - (\tau_i + f_i)^{k-j}$, thus two more passes are needed.
\end{proof}

\begin{algorithm}[H] 
\caption{4-Pass $F_k$ Computation} \label{alg:calculating moments}
\begin{algorithmic}[1] 
\State{Initial State: $r_o = 0 \; \forall i\in \{1,2, \ldots, n\},\; r_i = \tau_i, \; c_i= (-1)^{i}\binom{k}{i}, \forall j\in \{1, \ldots, k\}, r_j^i = \tau_j^i\; $} 
\State{ $\forall i\in [n]$, run $I_1^i$}
\State \textbf{for} every $x_\ell \in \bsig$ \textbf{do} $r_{x_\ell} \gets r_{x_\ell} + 1$ \Comment{1st pass}
\State $\forall i\in [n]$, run $I_2^i$
\State \textbf{for} every $x_\ell \in \bsig$ \textbf{do} $r_{x_\ell} \gets r_{x_\ell} - 1$ \Comment{2nd pass}
\State $\forall i \in [n]$, run $I_3^i$
\State\label{line:f_k comp two more pass line}Execute the inverse of every operation in reverse order  without changing $r_o$ using two more passes to restore $r_1,\ldots,r_n$.
\State \Return $r_o$
\end{algorithmic}
\end{algorithm}

\paragraph*{Polynomial Evaluation by Powering}
Alekseev et al.\ \cite{calcPolyCat25} show how to evaluate a multivariate polynomial of degree $k$ over $n$ variables by representing it as a linear combination of  powers ($\le k$) of linear functions of the variables (so called Waring representation \cite{landsberg2012tensors}). 
However, due to the large number of coefficients required for this representation, their program is designed for a specific polynomial. It does not receive a representation of the polynomial as input.
Their implementation requires catalytic memory of size exponential in $k$.
Similarly to our catalytic streaming implementation of the computation of $F_k$, we can implement this polynomial evaluation scheme in $4$  passes in our model. We will get a streaming algorithm for a \textbf{specific} multivariate polynomial $P$ that, given a stream,
evaluates $P(f_1,f_2,\ldots,f_n)$
 (it cannot get a polynomial as input).
Furthermore, it requires exponential in $k$ catalytic memory.

\section{Polynomial of Degree \texorpdfstring{$k$}{k} Using \texorpdfstring{$k+1$}{k+1} Passes}
\label{sec:poly}

In this section we show how to evaluate polynomials over the frequencies of items in a stream. We assume we get the polynomial's coefficients in our input tape, like we get $m$ and $n$. The following is this section's main result.

\begin{theorem}
\label{thm:polyk}
        Let $\bsig=x_1,x_2,\dots,x_m$ be a stream over $[n]$, and let $f_1,f_2,\dots,f_n$ denote the corresponding frequencies. Then there exists a catalytic streaming algorithm that for every multivariate polynomial $P \in \mathbb{F}[y_1,y_2,\dots,y_n]$ of total degree $k$,\footnote{$\mathbb{F}[y_1,y_2,\dots,y_n]$ is the ring of polynomials over the field $\F$ with the variables $y_1,y_2,\ldots,y_n$. The coefficients of $P$ are given to the algorithm on the input tape. }  evaluates \(
        P(f_1,f_2,\dots,f_n)
        \) in $k+1$ passes. This algorithm uses $O(1)$ regular registers and $n$ catalytic registers. It works over any ring $Z_q$ that contains inverses to $2,3,\dots,k$, to evaluate $P(f_1,f_2,\ldots,f_n)$ over the integers it is enough to take $q = O(\alpha \cdot m^k)$ where $\alpha \in \mathbb{Z}$ is the largest coefficient of any monomial.
\end{theorem}

\subsection{Warm Up: Degree 2 Polynomials}
To warm up, we first show how to evaluate degree $2$ polynomials. 

Let $\bsig = x_1, x_2, \dots, x_m$ be a stream of elements from the set $[n]$ with frequencies $f_1, f_2,\ldots f_n $ and let $P \in \mathbb{F}[y_1,y_2,\dots,y_n]$  be a multivariate polynomial of degree $2$. Let
\begin{equation*}
    P(f_1, f_2,\ldots, f_n) = \sum_{1\leq i\leq j \leq n} a_{i,j}f_i f_j \, ,
\end{equation*}
for some $\{a_{i,j}\} \subseteq \F$, i.e. we assume $P$ only contains degree $2$ monomials. Solving the problem with respect to such polynomials is easily generalized, since adding degree-1 monomials and constants is trivial using $O(\log m)$ bits of regular memory.

\medskip

We allocate registers $r_1, r_2,\ldots, r_n$ in the catalytic memory, each of size $2\log m$ bits. We denote the initial values inside these registers by $\tau_1, \tau_2, \ldots, \tau_n$. In addition, we allocate $2 \log m$ bits of regular memory as our output. We give this register a name, $r_o$, and think of it as the output register. The algorithm goes as follows.

\begin{enumerate}
        \item \label{item:d2 poly alg first set} Set $r_o \leftarrow P(r_1, r_2,\ldots, r_n)$.
        \item \label{item:d2 poly alg first pass} During the first pass of the stream, for every item $x_j$ received, increment the register $r_{x_j}$ by one. After this pass, $r_i$ holds $\tau_i + f_i$. 
        \item \label{item:d2 poly alg second set} Set $r_o \leftarrow r_o - 2 P(r_1, \ldots, r_n)$.
        \item During the second pass, we repeat Step~\eqref{item:d2 poly alg first pass}. At the end of this pass, $r_i$ holds $\tau_i + 2f_i$.
        \item \label{item:d2 poly alg third set} Set $r_o \leftarrow r_o + P(r_1, \ldots, r_n)$. 
        \item During the third pass, we reverse our previous operations, i.e. for every $x_j$ received, we subtract $2$ from its corresponding register. 
        \item Return $r_o /2$. 
    \end{enumerate}

To establish correctness, we follow the contents of $r_o$ along the execution of these steps. We use the following definitions.
\begin{align*}
    J &= \sum_{1\leq i\leq j \leq n} a_{ij} \tau_i \tau_j  &&(\text{the junk term})\, , \\
    C &= \sum_{1\leq i\leq j \leq n} a_{ij} (\tau_i f_j + f_i \tau_j) &&(\text{the cross term}) \, , \text{and} \\
    V &= \sum_{1\leq i\leq j \leq n} a_{ij} f_i f_j &&(\text{the target term}).
\end{align*}
We assume that the regular memory is initialized to $\0 =0^s$, hence initially $r_o  = 0$.
Following Step~\eqref{item:d2 poly alg first set},
\[r_o = P(\tau_1, \tau_2, \ldots, \tau_n) =  \sum_{1\leq i\leq j \leq n} a_{i,j}\tau_i \tau_j = J \, .\]
  Following Step~\eqref{item:d2 poly alg second set}, 
\begin{align*}
 r_o &= J - 2P(r_1, r_2, \ldots, r_n )= J - 2P(\tau_1 + f_1, \tau_2 + f_2, \ldots, \tau_n + f_n )  \\
 &  = J - 2\sum_{1\leq i\leq j \leq n} a_{i,j}(\tau_i + f_i)( \tau_j + f_j)=J -2(J+C+V) \, .
 \end{align*}
 Following Step~\eqref{item:d2 poly alg third set}, 
\begin{align*}
 r_o &  = J - 2(J+C+V) +\sum_{1\leq i\leq j \leq n} a_{i,j}(\tau_i + 2f_i)( \tau_j + 2f_j)\\
&  = J-2(J+C+V) + (J+ 2C+4V) = 2V
\end{align*}
It follows that when Step~\eqref{item:d2 poly alg third set} ends, $r_o = 2V = 2P(f_1,f_2,\ldots,f_n)$. Hence $r_o/2$ is exactly $P(f_1, f_2,\ldots, f_n)$, which is the desired output.

\subsection{Proof of \autoref{thm:polyk}}
In order to prove the theorem, we first handle the case of evaluating a single monomial.
\begin{lemma}
 \label{monomial-lemma}
    Let $\bsig=x_1,x_2,\dots,x_m$ be a stream over $[n]$, and let $f_1,f_2,\dots,f_n$ denote the corresponding frequencies. 
    There exists a catalytic streaming algorithm that given any monomial $M \in \mathbb{F}[y_1,y_2,\dots,y_n]$ of total degree $k$, evaluates $M(f_1,f_2,\ldots,f_n)$ in $k+1$ passes.
   The algorithm uses $O(1)$ regular registers and $n$ catalytic registers. This algorithm works over any ring $Z_q$ that contains inverses to $2,3,\dots,k$. To evaluate $M(f_1, \ldots, f_n)$ over the integers, it is enough to take $q = O(m^k)$.   
\end{lemma}
\begin{proof}
Let $r_1, r_2,\ldots, r_n$ be the catalytic registers, where each $r_i$ initially contains the value $\tau_i$, and let $r_o$ be the regular register.
Evaluating $M$ on $(f_1,f_2,\ldots,f_n)$ goes as follows (\autoref{alg:monomial}):

\begin{algorithm}[H]
\caption{Compute Degree-$k$ Monomial $M(f_1, f_2,\ldots,f_n)$ } \label{alg:calc monomial}
\label{alg:monomial}
\begin{algorithmic}[1]
\Input A monomial $M$ and a sequential stream $\bsig$. \Comment{\textbf{Initially}: $r_o = 0, \; \forall i, \; r_i = \tau_i$}
\Output The evaluation $M(f_1,f_2,\ldots,f_n)$.
\State $r_o \gets r_o + (-1)^kM(r_1,r_2, \ldots, r_n)$
\For{$i = 1, \dots, k$} 
\LineFor{every $x_j \in \bsig$}{$r_{x_j} \gets r_{x_j} + 1$}\Comment{The $i$-th pass}
    \State $r_o \gets r_o + (-1)^{k-i} \binom{k}{i} M(r_1, r_2, \ldots, r_n)$
\EndFor
\LineFor{every $x_j \in \bsig$}{$r_{x_j} \gets r_{x_j} - k$}\Comment{The $k+1$-th pass}
\State \Return $\frac{r_o}{k!}$
\end{algorithmic}
\end{algorithm}

Observe, that for every  $1 \leq i \leq k$ and $j \in [n]$, after the $i$-th pass, the value in register $r_j$ is $\tau_j + i \cdot f_j$. Thus, by summing the values added to $r_o$ in every pass, we get
\begin{equation} \label{stirling}
\begin{split}
    r_o &= (-1)^k M(\tau_1, \ldots, \tau_n) + \sum_{i=1}^k (-1)^{k-i} \binom{k}{i} M(\tau_1 + if_1, \ldots, \tau_n + i f_n) \\
    &=\sum_{i=0}^k (-1)^{k-i} \binom{k}{i} M(\tau_1 + i f_1, \ldots, \tau_n + if_n) \, .
\end{split}
\end{equation}
Since the exact degree of $M$ is $k$, we can choose $\{i_\alpha\}_{\alpha \in [k]} \subseteq [n]$ (these indices are not necessary distinct) such that $M(y_1,y_2,\ldots,y_n) = \prod_{\alpha = 1}^k y_{i_{\alpha}}$. Therefore, by expanding $M$ in Equation~\eqref{stirling} in that way we get
\begin{equation*}
    r_o = \sum_{i=0}^k (-1)^{k-i} \binom{k}{i} \prod_{\alpha=1}^k (\tau_{i_\alpha} + i f_{i_\alpha}) = \sum_{i=0}^k (-1)^{k-i} \binom{k}{i}\sum_{\ell = 0}^k i^{k-\ell}\text{Cross}_\ell(\btau,\boldsymbol{f}) 
\end{equation*}
where $\btau=\tau_1,\tau_2,\ldots,\tau_n$, $\boldsymbol{f}=f_1, f_2,\ldots,f_n$, and
\begin{equation*}
\text{Cross}_\ell(\btau, \mathbf{f}) =  \sum_{\substack{\beta =\{\beta_1, \ldots \beta_\ell\} \subseteq [k] \\ \gamma_1, \ldots, \gamma_{k-\ell = [k]\setminus \beta} }} \tau_{\beta_1}\cdots\tau_{\beta_\ell}\cdot f_{\gamma_1}\cdots f_{\gamma_{k-\ell}}
\end{equation*}
where the sum is over all partitions of $\{i_\alpha\}_{\alpha \in [k]}$ into two sets, $\{\beta_j\}_{j \in [\ell]}$ and $\{\gamma_t\}_{t \in [k -\ell]}$.

By changing the order of summation we get that

\begin{equation}
\label{eq:reorder}
r_o = \sum_{\ell = 0}^k\text{Cross}_\ell(\btau, \boldsymbol{f})\sum_{i=0}^k (-1)^{k-i} \binom{k}{i} i^{k-\ell}  \, .
\end{equation}

Recall the definition of Stirling numbers of the second kind:
\begin{equation*}
    S(n, k) = \frac{1}{k!} \sum_{i=0}^k (-1)^{k-i} \binom{k}{i} i^n \, .
\end{equation*}
A key property is that $S(n, k) = 0$ whenever $n < k$, and $S(k, k) = 1$.
Using these identities to rewrite Equation~\eqref{eq:reorder} we get
\begin{equation*}
    r_o = \sum_{\ell=0}^k \text{Cross}_{\ell}(\btau,\boldsymbol{f}) \cdot k!\cdot S(k-\ell,k) = k! \cdot \text{Cross}_0(\btau,\boldsymbol{f}) = 
 k! \cdot \prod_{\alpha=1}^k  f_{i_\alpha} = k! \cdot M(f_1,f_2,\ldots,f_n) \, .
\end{equation*}
Therefore $\frac{r_o}{k!} = M(f_1,f_2,\ldots,f_n)$ and the returned value is correct.
\end{proof}

We are now ready to prove \autoref{thm:polyk}.

\begin{proof}[Proof of \autoref{thm:polyk}]
 Denote the set of monomials with nonzero coefficient in $P(\bm{y})$ by $\text{Mon}(P)$ and let $a_M \in Z_q$ be the coefficient of the monomial $M(\y)$ in $P(\y)$. Moreover, let $r_1,r_2,\ldots,r_n$ be catalytic registers and $r_o$ a regular register. Algorithm $3$ shows how to evaluate $P$ on $f_1,f_2,\ldots,f_n$.

\begin{algorithm} [H]
\caption{Compute Degree-$k$ Polynomials}
\label{alg:polyk}
\begin{algorithmic}[1]
\Statex \Input A polynomial $P$ and a sequential stream $\bsig$. \Comment{\textbf{Initially}: $r_o = 0, \; \forall i, \; r_i = \tau_i$}
\Output The evaluation $P(f_1,f_2,\ldots,f_n)$.
\For{$i = 1, \dots, k$} \Comment{The $i$-th pass}
    \LineFor{every $x_j \in \bsig$}{$r_{x_j} \gets r_{x_j} + 1$}
    \For{$M \in \text{Mon}(P)$}
            \State $r_o \gets r_o + a_M(-1)^{\deg(M)-i} \binom{\deg(M)}{i} M(r_1, r_2,\ldots, r_n) \; \big/ \,\deg(M)!$
    \EndFor
\EndFor
\LineFor{every $x_j \in \bsig$}{$r_{x_j} \gets r_{x_j} - k$} \Comment{The $k+1$-th pass}
   \For{$M \in \text{Mon}(P)$}
        \State $r_o \gets r_o + a_M(-1)^{\deg(M)}  M(r_1, \ldots, r_n) \; \big/ \,\deg(M)!$
    \EndFor
\State \Return $r_o$
\end{algorithmic}
\end{algorithm}

The idea is to evaluate in parallel all monomials of $P$ using the same set of registers. Since summation is commutative different monomials do not interfere with one another. The correctness of \autoref{alg:polyk} essentially follows as in the proof of \autoref{monomial-lemma}.
\end{proof} 


\subsection{Polynomial Evaluation over a Field Using Primitives Roots of Unity}

In Cook and Mertz~\cite{CMtree24}, the authors show how to evaluate any multivariate polynomial $P(y_1,\ldots,y_n)$ of degree $k$ over a finite field $\F_q$  under the assumption that $k +1 < q$. 
Their algorithm requires a primitive root of unity of 
order $e \ge k+1$. Note that the smallest $e$ for which such a root exists is the smallest divisor of $q-1$ which is larger than $k$. Their algorithm reads each input variable $e$ times, and uses $\log q$ bits of regular memory together with $n\log q$ bits of catalytic memory. As we did for $F_k$ we can modify this algorithm to a catalytic streaming algorithm that evaluates $P(f_1,\ldots, f_n)$ in $e\ge k+1$ passes over the stream. To evaluate $P(f_1,\ldots, f_n)$ over the integers we have to use a field of characteristic $\Omega(m^k)$. Note that it is not clear how to find a field of characteristic $O(m^k)$ that contains a primitive root of unity of order $k+1$. Hence this algorithm may require more than $k+1$ passes. The algorithm presented in this section, which is arguably simpler, always evaluates $P$ in exactly $k+1$ passes and requires only milder assumptions on the underlying ring.

Goldreich \cite{goldreichTreeEval} simplified the Cook-Mertz procedure.
He observed that Cook and Mertz in fact interpolate the univariate polynomial
$f(x)=P(x\cdot \tau_1+f_1,\ldots,x\cdot\tau_n + f_n)$ at powers of a primitive root of unity, but one can also interpolate at other points. (Since we only need to extract $f(0)$ at the end.)
This removes the need for a primitive root of unity, however, when we attempt to implement it in the catalytic streaming model it would require twice as many passes. This  happens
since to evaluate at arbitrary values of $x$ requires two passes per value (we need to subtract away the frequencies $f_i$, multiply the ``noise'' by a different value of $x$ and then add back the frequencies).  But powers of a primitive root of unity allow to get away with one pass per value of $x$.


\begin{remark}
    Using this method we get a $3$-pass algorithm computing $F_2$ using $O(\log m)$ bits of regular memory.  
\end{remark}

\begin{remark}
    In this section we have evaluated polynomials over the frequencies of the elements in the stream. Note that this can be easily generalized to evaluating $\ell$-variate polynomials over every set of functions of the steam elements $\{g_1(\bsig), \ldots, g_\ell(\bsig)\}$, such that we can compute $r_i\leftarrow g_i(\bsig)$ for every $i$ in one pass and within our space bounds.
\end{remark}

\section{Lower Bounds} \label{sec:lower bounds}
Standard streaming lower bounds are typically proved via reductions from communication complexity. In this section, we apply this approach to the catalytic setting. Specifically, in \autoref{sec:imp}, we leverage the recently introduced framework of catalytic communication complexity~\cite{catalyticCommunication} to establish lower bounds for the catalytic streaming model.

However, catalytic communication results cannot establish lower bounds for multi-pass catalytic streaming. As we show below, this is because a modified inner product protocol from \cite{catalyticCommunication} actually solves \textsc{Set-Disjointness}. Using alternative techniques, we instead prove that our 3-pass algorithm for $F_2$ is pass-optimal among a family of ``natural'' algorithms.

\subsection{Implication of Catalytic Communication}
\label{sec:imp}
In this section, we use results from~\cite{catalyticCommunication} to establish a lower bound for one-pass catalytic streaming algorithms. We then adapt their method for computing inner product over $\text{GF}_2$ to obtain a catalytic communication protocol for \textsc{Set-Disjointness} in three rounds of communication that uses only $O(\log n)$ bits of regular memory. This latter result shows that the typical approach to prove streaming lower bounds via reductions from \textsc{Set-Disjointness} cannot yield lower bounds in our new model for algorithms with more than one pass.

Briefly, the communication model of \cite{catalyticCommunication} is as follows (for the exact definition, see \cite[Definition~$1$]{catalyticCommunication}).
Alice, holding an input $\x\in \{0,1\}^n$, and Bob, holding an input  $\y\in \{0,1\}^n$ communicate to compute a function $f(\x,\y)$. They exchange messages on $c$  bits of catalytic memory (arbitrarily initialized) and $s$ bits of regular memory (initialized to $\0$), note that the message size in every round is fixed.  In their protocol the last player sending a message should send/restore the original contents of the catalytic memory. The player that gets this message should output $f(\x,\y)$.

A lower bound on the amount of regular memory required by a catalytic streaming algorithm for computing the frequency moment $F_k$ follows from the following three facts: (1) We observe that the standard reduction from the \textsc{Set-Disjointness} problem in communication complexity to the streaming problem of computing the frequency moment also works in the catalytic setting. 
The \textsc{Set-Disjointness} problem in communication complexity asks how many bits of communication Alice, holding $x\in \{0,1\}^n$,  and Bob, holding $y\in \{0,1\}^n$,  have to exchange in order to compute the function 
\begin{equation*}
    \disj_{n}(x,y) = \begin{cases} 1 \hspace{15pt} \forall i \in [n]: x_i =1 \Rightarrow  y_i = 0  
        \\
        0 \hspace{15 pt} \text{otherwise}
    \end{cases} \; .
\end{equation*}

Specifically, our lower bound relies on three facts:
(1) A $p$-pass catalytic streaming algorithm for a frequency moment $F_k$ yields a $2p$-round catalytic communication protocol for \textsc{Set-Disjointness} ($\disj_n$) with the same catalytic memory and $O(\log n)$ regular memory (a standard reduction provided in \autoref{app:dis}).
(2) Any catalytic communication protocol with fewer than three rounds can be simulated without catalytic memory at a similar regular communication cost \cite[Proposition~3]{catalyticCommunication}.
(3) Computing $\disj_n$ strictly requires $\Omega(n)$ standard communication \cite{kal92setDisj, razborovSetDisj, AMS}. Combining these, we immediately obtain the following:

\begin{proposition}
There is no one-pass catalytic streaming algorithm for frequency moments with sublinear regular memory.
\end{proposition}

\begin{proof}
Assume there exists a 1-pass catalytic streaming algorithm with sublinear regular memory. By Fact (1), this yields a 2-round catalytic protocol for $\disj_n$ with sublinear regular communication. By Fact (2), this implies a standard 2-round protocol for $\disj_n$ with sublinear communication, which contradicts Fact (3).
\end{proof}

While similar reductions can rule out other one-pass catalytic streaming algorithms, this approach fails for multiple passes (which induce protocols with more than two rounds). This barrier is inherent: in \autoref{app:inner product}, we show that a simple modification of the $\text{GF}_2$ inner-product protocol from \cite[Proposition~5]{catalyticCommunication} yields a 3-round catalytic protocol for $\disj_n$. Thus, multi-pass lower bounds cannot be obtained via this route.

\subsection{A Restricted Lower Bound for Moments in Two Passes}\label{sec:restricted LB}

In this section, we characterize a set of properties of catalytic streaming algorithms, which seem natural, such that every algorithm that have these properties cannot calculate $F_k$, for $k \ge 2$, in two passes.

\begin{assumption}[Restricted two-pass Catalytic Streaming Algorithm]
\label{restricted}
We denote by $C$ the catalytic memory, $r_i$ is a catalytic register associated with item $i$, and $M$ is the regular memory.
A restricted two-pass catalytic streaming algorithm obeys the following structure: 

\begin{enumerate}
\item \textbf{Initial Update (Before we see the stream)}:  $M \leftarrow M + W_s(\bm{C})$, where $W_s$ is some function of the entire catalytic memory $\bm{C}$. It returns a vector that we use in order to update the clean memory $M$ additively.
\item \textbf{Pass 1 (Forward):} For each item $x_j$ in the stream, if $x_j=i$ then we set $r_{i} \leftarrow U_{i}(r_i)$, and then $M \leftarrow M + P_i(r_{i})$.
Here $U_i$ is   an item-specific bijection $U_i : \Sigma \rightarrow \Sigma$.
    We require that $r_i \not= r_j$ for $i\not= j$. $P_i$ is also an item specific function that updates $M$ additively after the change to $r_i$.
\item \textbf{Intermediate Update (End of Pass 1):}  $M \leftarrow M + W_m(\bm{C})$. $W_m$ is  a function of the entire catalytic memory $\bm{C}$
that we use to update the clean memory between the two passes.

\item \label{item:2 pass lb assump backward pass} \textbf{Pass 2 (Backward):} To restore the catalytic memory, if $x_j=i$, we set $r_{i} \leftarrow U_i^{-1}(r_{i})$, then $M \leftarrow M + Q_i(r_{i})$.
$Q_i$ is an item specific function that updates $M$ additively after the update to  $r_i$.

\item \textbf{Final Update (End of Pass 2):}  $M \leftarrow M + W_e(\bm{C})$. $W_e$ is  a function of the entire catalytic memory $\bm{C}$
that we use to update the clean memory between following the two passes.
\end{enumerate}
\end{assumption}

\begin{remark}
\autoref{restricted}
 can be generalized to apply to $p$-pass streaming algorithms. If $U_{i,j}$ is the function applied to $r_i$ in pass $j$ when we see item $i$, then we should have 
$U_{i,p}^\ell(\cdots (U_{i,2}^\ell(U_{i,1}^\ell(\tau_i))))=\tau_i$ for every $\ell \leq m$. 
Further generalizations of this assumption  
allocate a set of registers to every element instead of only one register and allow to apply some function $g_{i,j}$ to each $r_i$ after pass $j$ such that
$g_{i,p}(U_{i,p}^\ell(\cdots (g_{i,2}(U_{i,2}^\ell (g_{i,1}(U_{i,1}^\ell(g_{i,0}(\tau_i))))))))=\tau_i$ for every $\ell \leq m$. All our algorithms obey this generalized assumption.
\end{remark}

Now we show that no program that satisfies \autoref{restricted} can compute $F_2$ exactly using only two passes. 
The proof for $F_k$, $k> 2$ is analogous. The intuition is that the catalytic tape behaves as a one time pad with respect to the items frequencies, so, in one pass it is impossible to extract information about them from the catalytic memory.
Formally, we prove this by analyzing the 
discrete derivative of the output as a function of $f_i$.

\subsubsection{The Discrete Derivative Argument}  
Let $\bm{C} = \bm{\mathcal{T}}$ and in particular
let $r_i = \tau_i$ initially for all $i\in [n]$. For every $i \in [n]$ and $j \in [f_i]$ we denote by $U_i^j$ the bijective function $U_i$ applied $j$ times. We analyze the state of $M$ at the end of the algorithm, denoted by $M_e(\btau,\boldsymbol{f})$. 
\[
M_e(\btau,\boldsymbol{f}) = W_s(\bm{\mathcal{T}})+ \sum_{i=1}^n \sum_{j=1}^{f_i} P_i\big(U_i^j(\tau_i)\big) + W_m\big(\boldsymbol{U^{f}(\btau)}\big) +  \sum_{i=1}^n\sum_{j=0}^{f_i-1} Q_i\big(U_i^j(\tau_i)\big) + W_e(\bm{\mathcal{T}})
\] 

where $\boldsymbol{U^{f}(\btau)} = \left(U^{f_1}_1(\tau_1),\ldots, U_n^{f_n}(\tau_n)\right)$. The second term from the right comes from the fact that in the second pass (see Item~\ref{item:2 pass lb assump backward pass}) after seeing item $i$ for the $j$-th time, the value of the register $r_i$ is $U_i^{f_i-j}(\tau_i)$.

For simplicity, for the rest of the proof we denote \[ E_i(\tau_i,f_i) =  \sum_{j=1}^{f_i} P_i\big(U_i^j(\tau_i)\big) + \sum_{j=0}^{f_i-1} Q_i\big(U_i^j(\tau_i)\big)\]

\begin{theorem}

 There is no algorithm satisfying \autoref{restricted}  for which $M_e(\btau,\boldsymbol{f}) = \sum_{i=1}^n f_i^2 $.\footnote{We assume w.l.o.g.\ that in the end of the algorithm we only have $F_2$ on our regular memory.}

\end{theorem}

\begin{proof}
Assume, by way of contradiction, that for every initial state $\bm{\mathcal{T}}$ of the catalytic tape and for every set of $f_i$'s we have an algorithm such that
\begin{equation}
\label{master}
M_e(\btau,\boldsymbol{f}) = \sum_{i=1}^n E_i(\tau_i,f_i) + W_s(\bm{\mathcal{T}})+ W_m\big(\boldsymbol{U^{f}(\btau)}\big)     + W_e(\bm{\mathcal{T}})= \sum_{i=1}^n f_i^2 
\end{equation}

We take the discrete partial derivative of $M_e$ with respect to $f_i$, which in our case is defined as $\Delta_{f_i}(g) \coloneqq g(f_1,\ldots,f_i + 1, \ldots,f_n) - g(\boldsymbol{f})$, of both sides of Equation~\eqref{master}.
On the right hand side we get:

\[
\Delta_{f_i} \left( \sum_{j=1}^n f_j^2 \right) = (f_i + 1)^2 - f_i^2 = 2f_i + 1 \, .
\]

For the left hand side first observe that
\begin{equation*}
    \Delta_{f_i} \left(\sum_{i=1}^n E_i(\tau_i,f_i)\right) = P_i\big(U_i^{f_i+1}(\tau_i) \big) + Q_i\big(U_i^{f_i}(\tau_i) \big) =  P_i\big(U_i(y) \big) + Q_i\big(y\big) \ ,
\end{equation*}
where $y = U_i^{f_i}(\tau_i)$ is the state of register $r_i$ at the end of the forward pass. Note that $\Delta_{f_i}\big(W_s(\btau)\big)=\Delta_{f_i}\big(W_e(\btau)\big) = 0$. 
Overall, we get

\begin{align*}
\Delta_{f_i}\big(&M_e(\btau,\boldsymbol{f})\big) = P_i\big(U_i(y) \big) + Q_i\big(y\big)  + \\ &W_m\left(U^{f_1}_1(\tau_1),\ldots,\bm{U_i(y)},\ldots, U_n^{f_n}(\tau_n)\right) - W_m\left(U^{f_1}_1(\tau_1),\ldots,\bm{y},\ldots, U_n^{f_n}(\tau_n)\right) \ .
\end{align*}
We argue that
$\Delta_{f_i}\big(M_e(\btau,\boldsymbol{f})\big)
= 2f_i + 1$ is not possib which would give a contradiction.

We fix $f_j$ and $\tau_j$ for $j \neq i$, and think of  $\Delta_{f_i}\big(M_e(\btau,\boldsymbol{f})\big)$ as a function of $y = U_i^{f_i}(\tau_i)$, which we denote by $H_i(y)$. Note that the dependency of $\Delta_{f_i}\big(M_e(\btau,\boldsymbol{f})\big)$ on $f_i$ is only through $y = U_i^{f_i}(\tau_i)$. Let \(\tau_i\) and \(f_i\) be random variables, each distributed uniformly over its set of possible values. The intuition is as follows. Since $\tau_i$ is a uniform random variable and $U_i$ is a bijection it follows that
$y$ is also distributed uniformly at random and independent of $f_i$. Therefore it cannot hold any information about $f_i$. 
This is formalized in the following claim.

\begin{claim}
\label{information}
    Let $f_i$ and $\tau_i$ be independent random variables over the uniform distribution and let $I(\cdot,\cdot)$ be the mutual information function. Then $I(f_i,y) = 0$.
\end{claim} 
\begin{proof}
By \autoref{restricted}, $U_i$ is a bijection, thus it holds that $\forall f_i, \; U_i^{f_i}: \Sigma\rightarrow \Sigma$ is also a bijection. Therefore we have that $y$ is a random variable that is distributed uniformly over $\Sigma$. In particularly he have \[\text{Pr}_{\tau_i\sim U_{\Sigma}}\left[y = \sigma \mid f_i = k\right] = \text{Pr}_{\tau_i\sim U_{\Sigma}}\left[U^k(\tau_i)=\sigma\right]= \frac{1}{|\Sigma|}\] Where the last equality follows from the fact that $U_i^k(\cdot)$ is a bijection. 
Since the conditional probability of $y$ is completely independent of the choice of $f_i$ we have that $I(f_i,y) = 0$.    
\end{proof}

Since $H_i(y)$ is a deterministic function of $y$ we have that
$I(f_i,H_i(y))\le I(f_i,y) = 0$ by the data processing inequality and \autoref{information}. We have assumed the algorithm always computes $F_2$ exactly, therefore $H_i(y) = 2f_i +1$, which implies that $I\big(f_i,H(y)\big) = 1$ in contradiction.
\end{proof}

The last proof works with small changes for every $F_k$ with $k \geq 2$.

In \autoref{proof break} we show where the proof fails for $3$ passes.

\section{Computing $F_2$ in two passes}

Having identified the limitations of the family of algorithms considered in section \ref{sec:restricted LB}  we took a different approach and construct in this section a two-pass algorithm with the following properties.

\begin{theorem}[Two-pass catalytic computation of $F_2$]
\label{thm:two-pass-F2}
Let $\bsig=x_1,\ldots,x_m$ be a stream over $[n]$, and let
$f_i=\left|\{t\in[m]:x_t=i\}\right|$ be the frequency of item $i$.
There is a two-pass catalytic streaming algorithm that computes
$F_2(\bsig)=\sum_{i=1}^n f_i^2$ 
exactly using $O(1)$ regular registers and $O(nm)$ catalytic registers over the ring $Z_q$ with $q = \Theta(m^2)$, and hence uses
$O(nm\log m)$ bits of catalytic memory and $O(\log(nm))$ bits of regular memory. 
\end{theorem}

\begin{proof}
For simplicity, we assume $m$ is a power of $2$; the algorithm easily generalizes to any value of $m$.

We first describe the algorithm. 
Consider the complete binary tree $T$ with $m$ leaves, where the $t$-th leaf from the left is identified with the stream element $x_t$ at position $t$.
For every internal node $v$ of $T$, and every item $i\in[n]$, we allocate one catalytic register $r_{v,i} \in Z_q$.
We denote the arbitrary initial value of $r_{v,i}$ by $\tau_{v,i}$.
The algorithm also maintains one regular output register $r_o$, initialized to $0$.

The idea is that each internal node $v$ counts the number of pairs of equal stream elements $x_s=x_t=i$, for $s < t$, such that $x_s$ is in the left subtree of $v$ and $x_t$ is in the right subtree of $v$.
Each such pair of equal stream elements has a unique lowest common ancestor in $T$. Summing these contributions over all internal nodes counts every pair of equal stream elements exactly once.

The algorithm proceeds as follows.

\medskip
\noindent
\textbf{Pass 1 (Forward):}
When $x_t = i$ arrives, we scan all ancestors $v$ of the leaf corresponding to $t$, and for each we do the following:
\begin{itemize}
    \item If we arrived from the left child of $v$, we set $r_{v,i}\gets r_{v,i}+1$.
    \item If we arrived from the right child of $v$, we set $r_o\gets r_o+r_{v,i}$.
\end{itemize}

\noindent
\textbf{Pass 2 (Backward):}
Again, when $x_t = i$ arrives, we scan all ancestors $v$ of the leaf corresponding to $t$, and for each we do the following:
\begin{itemize}
    \item If we arrived from the left child of $v$, we set $r_{v,i}\gets r_{v,i}-1$.
    \item If we arrived from the right child of $v$, we set $r_o\gets r_o-r_{v,i}$.
\end{itemize}
    
At the end of the second pass, the algorithm outputs $m+2r_o \pmod q$.
We give pseudo-code in \autoref{alg:binary-tree-f2}.

We solve the case where $m$ is not a power of $2$ by adding ``dummy'' leaves to the tree up to the next power of $2$.

We now prove correctness. 
Fix an internal node $v$ with a left child $v_\ell$ and a right child $v_r$.
Let $f_i(v_\ell)$ and $f_i(v_r)$ be the number of occurrences of item $i$ in the subtrees of $v_\ell$ and $v_r$, respectively.
Because the stream is processed chronologically from left to right, all leaves in $v_\ell$ are processed before any leaves in $v_r$.

Assume first that the initial state is $\tau_{v,i}=0$ for all catalytic registers $r_{v,i}$.
During Pass 1, when we process the left subtree $v_\ell$, we increment $r_{v,i}$ exactly $f_i(v_\ell)$ times. 
Thus, when we begin processing the right subtree $v_r$, the register contains $r_{v,i} = f_i(v_\ell)$.
Each time we encounter $i$ in $v_r$ (which happens $f_i(v_r)$ times), we add $r_{v,i}$ to $r_o$.
Therefore, while processing the stream underneath $v$, we add exactly $f_i(v_\ell)f_i(v_r)$ to $r_o$.
This is exactly the number of pairs of $i$'s where one occurs in $v_\ell$ and the other in $v_r$.
Since each pair of equal stream elements $x_s=x_t=i$ (with $s < t$) has a unique lowest common ancestor, it contributes exactly once to the sum in $r_o$.

Since this holds for every item $i\in [n]$, at the end of the first pass we get:
\[
r_o = \sum_{i=1}^n \binom{f_i}{2}.
\]

The algorithm outputs
\[
m+2r_o = \sum_{i=1}^n f_i + 2\sum_{i=1}^n \frac{f_i^2-f_i}{2} = \sum_{i=1}^nf_i^2 = F_2(\bsig).
\]

Because $q = \Theta(m^2)$ is chosen such that $q>m^2$ and $F_2(\bsig)\le m^2$, this equality over $Z_q$ determines the exact integer value of $F_2(\bsig)$.

Now we argue that we get the exact same result even when the initial values $\tau_{v,i}$ are arbitrary (i.e., there is initial ``garbage'' in the registers).
During the first pass, when processing the left subtree $v_\ell$, we again increment $r_{v,i}$ exactly $f_i(v_\ell)$ times. Thus, when we process the right subtree $v_r$, the register $r_{v,i}$ contains $\tau_{v,i} + f_i(v_\ell)$. We add this value to $r_o$ exactly $f_i(v_r)$ times, contributing $f_i(v_r)(\tau_{v,i} + f_i(v_\ell))$ to $r_o$.

Crucially, the second pass also processes the stream in the forward chronological direction. When processing the left subtree $v_\ell$ in the second pass, we decrement $r_{v,i}$ exactly $f_i(v_\ell)$ times. Therefore, by the time we begin processing the right subtree $v_r$ in the second pass, $r_{v,i}$ has been perfectly restored to its initial value $\tau_{v,i}$. Thus, during the processing of the right subtree $v_r$ in the second pass, we subtract $\tau_{v,i}$ from $r_o$ exactly $f_i(v_r)$ times. 

The arbitrary value $\tau_{v,i}$ perfectly cancels out across the two passes:
\[
f_i(v_r)(\tau_{v,i} + f_i(v_\ell)) - f_i(v_r)\tau_{v,i} = f_i(v_\ell)f_i(v_r).
\]
This leaves a net addition of $f_i(v_\ell)f_i(v_r)$ to $r_o$, exactly as in the zero-initialized case. Furthermore, since we increment $r_{v,i}$ in the first pass and decrement it under the exact same conditions in the second pass, every catalytic register is restored to its initial value $\tau_{v,i}$ by the end of the computation.

Finally, we analyze the complexity of the algorithm. The complete binary tree has $m-1=O(m)$ internal nodes. For each internal node $v$ and item $i\in[n]$, we use one catalytic register over $Z_q$, so we have $O(nm)$ registers in total. They consist of $n(m-1)\lceil\log q\rceil=O(nm\log q)$ bits.
Substituting $q= \Theta(m^2)$, we get a total of $O(nm\log m)$ catalytic bits.

We need one regular register of $O(\log m)$ bits for $r_o$, and we also use regular memory to address the catalytic registers (to navigate between nodes of $T$). Addressing $O(nm)$ catalytic registers requires $O(\log(nm))$ bits.

To process each stream position, the algorithm loops over the $O(\log m)$ ancestors of the corresponding leaf, so the total running time is $O(m\log m)$ arithmetic operations per pass. The algorithm uses exactly two passes over the stream.
\end{proof}

\begin{algorithm}[H]
\caption{Two-Pass Catalytic Computation of $F_2$ via a Binary Tree}
\label{alg:binary-tree-f2}
\begin{algorithmic}[1]
\Statex \Input A sequential stream $\bsig=(x_1,\ldots,x_m)\in[n]^m$.
\Statex \Output The evaluation $F_2(\bsig)$.
\State Let $T$ be the complete binary tree with $m$ leaves.
\State Identify the leaves of $T$ from left to right with the stream positions $1,\ldots,m$.
\State For every internal node $v$ of $T$ and every item $i\in[n]$, allocate a catalytic register $r_{v,i}$. 
\Comment{\textbf{Initially}: $r_{v,i}=\tau_{v,i}$}
\State Initialize the regular output register $r_o\gets 0$.

\medskip
\For{$t = 1, \dots, m$}   \Comment{The 1st pass}
    \State $u\gets$ the leaf of $T$ corresponding to position $t$
    \While{$u$ is not the root}
        \State $v\gets \mathrm{parent}(u)$
        \If{$u$ is the left child of $v$}
            \State $r_{v,x_t}\gets r_{v,x_t}+1$
        \Else
            \State $r_o\gets r_o+r_{v,x_t}$
        \EndIf
        \State $u\gets v$
    \EndWhile
\EndFor

\medskip
\For{$t = 1, \dots, m$}   \Comment{The 2nd pass}
    \State $u\gets$ the leaf of $T$ corresponding to position $t$
    \While{$u$ is not the root}
        \State $v\gets \mathrm{parent}(u)$
        \If{$u$ is the left child of $v$}
            \State $r_{v,x_t}\gets r_{v,x_t}-1$
        \Else
            \State $r_o\gets r_o-r_{v,x_t}$
        \EndIf
        \State $u\gets v$
    \EndWhile
\EndFor

\State \Return $m+2r_o \pmod q$
\end{algorithmic}
\end{algorithm}

\section{Applications}
\label{sec:app}



In this section, we apply our polynomial evaluation algorithm to design multi-pass streaming algorithms for the exact computation of quantities that are notoriously difficult to evaluate exactly in the standard streaming model. Specifically, we present algorithms for counting distinct elements, counting occurrences of small subgraphs in an edge stream, and identifying the exact set of frequent elements.

\subsection{$F_0$ Using Powering}

The zero moment of a stream  $\bsig = x_1, x_2, \ldots, x_m$ over $[n]$, is $F_0(\bsig) := \sum_{i=1}^n \bm{1}_{f_i >0}$; that is the number of distinct items that appear in $\sigma$. 
Computing $F_0$ in the standard streaming model  using a constant number of passes is known to be hard. In particular, \cite{AMS} show that any $p$-pass algorithm  requires $\Omega(min\{m,n\}/p)$ bits of memory. In contrast, in the catalytic streaming model we have the following Lemma.

\begin{lemma}[$F_0$ Computation]
Let $\bsig=x_1,x_2,\dots,x_m$ be a stream over $[n]$. There exists a catalytic streaming algorithm that evaluates
\(
F_0(\bsig)
\)
in $4$ passes, using one regular register and $n$ catalytic registers. This algorithm works over the field $\mathbb{F}_p$ for every prime $p > m$ .
\end{lemma}
\begin{proof}
By Fermat's little theorem we have $f_i^{p-1} \equiv 1 \pmod p$ for every $f_i \not \equiv 0 \pmod p$, and $0$ whenever $f_i \equiv 0 \pmod p$. Therefore, evaluating the $p-1$-th moment of $\bsig$ over the field $\mathbb{F}_p$ gives us $F_0(\bsig)$.
\end{proof}

\begin{remark}
    Note that by Bertrand's postulate there exists a prime $m < p \le 2m$, so the registers used require at most $\log (m) + 1$ bits.
\end{remark}

\subsection{Finding Subgraphs in a Stream Using Polynomials}

Consider a multigraph $G = (V, E)$ with $|V| = n$ and $|E| = m$. In the graph streaming model, $G$ is revealed sequentially as a stream of edges $e_1, e_2, \ldots, e_m$. A fundamental problem in this setting is computing the number of occurrences  of a small target subgraph $H$ (e.g., a triangle or four cycle) within $G$, or simply detecting its presence. This task is difficult because local structural information is fragmented across the stream. In fact, exact counting, or even detecting a single instance of a simple subgraph like a triangle, requires $\Omega(m/p)$ bits of memory in the worst case for any $p$-pass streaming algorithm \cite{BarYossef2002,bera_et_al:LIPIcs.STACS.2017.11, braverman2013hard}.



\begin{lemma}[Subgraph Counting]
Let $H = (V_H, E_H)$ be a fixed target multigraph. There exists an algorithm that, given an edge stream of length $m$ over a vertex set $V$ of size $n$, computes the exact number of occurrences of $H$ in the underlying graph $G$. The algorithm requires $|E_H|+1$ passes over the stream, uses $O(|E_H| + |V_H|)$ regular registers and $n^2$ catalytic registers, and works over $Z_q$ with $q = O\big(\max\{m^{|E_H|}, n\}\big)$ .
\end{lemma}

\begin{proof}
For simplicity, consider first the case where $H$ is a triangle.

\begin{observation}
Denote by $x_{u,v}$ the number of occurrences of the edge $(u,v)$ in $G$. The exact number of triangles in $G$ is therefore given by the polynomial:
\[ P(\{x_{u,v} \mid (u,v)\in E\} )= \sum_{u < v < w \in V} x_{u,v} \cdot x_{v,w} \cdot x_{u,w} \, . \] 
\end{observation}

Given this observation, we can compute the number of triangles by evaluating $P$ over the stream using \autoref{alg:polyk} for cubic polynomials. 

We generate these coefficients of the polynomial $P(\{x_{u,v} \mid (u,v)\in E\})$ on the fly. To do so, we allocate $3$ registers in our regular memory, to index the vertices $u$, $v$, and $w$. By incrementally updating these registers, we iterate through all triplets $(u, v, w)$ such that $u < v < w$. For each such triplet, the corresponding monomial $x_{u,v} \cdot x_{v,w} \cdot x_{u,w}$ appears in $P$ with coefficient $1$.


\paragraph*{Generalization to Arbitrary Subgraphs}
To generalize this approach to an arbitrary subgraph $H = (V_H, E_H)$, we define a polynomial of degree $k = |E_H|$.  For every set of $|V_H|$ vertices and assignment of labels to the vertices we generate a monomial. This generates each monomial as many times as the number of automorphisms of $H$. So we divide the result by the cardinality of the automorphism group of $H$. This number should either be given as input or if $H$ is small it can be computed in (regular) memory polynomial in the size of $H$, by checking all permutations of the vertices in $H$ lexicographically.
We evaluate this polynomial over the frequencies of edges using \autoref{alg:polyk}.
As for triangles we do not store the polynomial but traverse the monomials using $|V_H|$ registers of size $\log n$ each.
Applying \autoref{thm:polyk}, this takes $|E_H|+1$ passes, and calculations are over some ring which requires registers of size at most $2 \Big( \frac{m}{|E_H|}\Big)^{|E_H|}$ bits since $\Big( \frac{m}{|E_H|}\Big)^{|E_H|}$ is an upper bound on the number of occurrences of $H$ in a multigraph with $m$ edges. 
\end{proof}

\subsection{ $F_2$ Heavy Hitters}
We use our two-pass algorithm for $F_2$ to compute 
 $\epsilon F_2$-heavy hitters. We say that an element $i$ is $\epsilon F_2$-heavy hitter if $f_i^2 \ge \epsilon F_2$.
We do it in logarithmic number of passes and with $O(1/\epsilon)$ registers of clean memory.  This is about the same amount of memory required by the classical state of the art (one-pass) streaming algorithm for $F_2$-heavy hitters \cite{braverman2017bptree}, just that we identify them exactly rather than approximately using a simple divide an conquer scheme.\footnote{Maybe a more common notion is $\epsilon$ $\ell_2$-heavy hitter, defined to be any item $i$ such that $f_i\ge \epsilon \sqrt{F_2}$. These notions are the same up to squaring $\epsilon$.}
Specifically we show the following.

\begin{lemma}
    There is an algorithm that given a stream $\bsig = x_1, x_2, \ldots, x_m$, returns the exact set of $\epsilon F_2$-heavy hitters of $\bsig$ in $2+2\log \epsilon n$ passes using $O(1/\epsilon)$ registers of regular memory and $O(nm)$ registers of catalytic memory,  over a ring of size $Z_q$ where $q = O(m^2)$ .
\end{lemma}

\begin{proof}
\looseness=-1
The formal description of the algorithm is given in  Algorithm~\ref{alg-hh}.
In short we maintain a set
$\mathcal{B}$ of at most $1/\epsilon$ buckets, each is an interval of consecutive elements from $[n]$, containing the heavy hitters. In each iteration we partition each bucket into two buckets. Then we estimate the contribution to $F_2$ of the elements in each bucket simultaneously by running a copy of   \autoref{alg:binary-tree-f2} separately on the elements of each bucket. 
There could be at most $1/\epsilon$ buckets whose contribution is more than $\epsilon F_2$, we continue with them to the next iteration and discard the rest.
We need two passes to compute $F_2$ initially and then two passes in each iteration to compute the $F_2$ contributions of the buckets. The number of iterations is at most $\log \epsilon n $, as the initial size of the buckets is $\epsilon n$.
We keep in the regular memory the boundaries of the buckets in
$\mathcal{B}$ and the $F_2$
values of the elements in each of at most $2/\epsilon$ buckets.
Thus we need $O(1/\epsilon)$ registers over $Z_q$ for $q = O(m^2)$.

 The catalytic memory we need is the same as required by  \autoref{alg:binary-tree-f2}. In order to evaluate the $F_2$ of the elements of every bucket we use a tree that needs at most $O(\epsilon n m \log m)$ bits of catalytic memory. At every iteration we evaluate $1/\epsilon$ such buckets. Therefore the total number of bits in the catalytic memory the algorithm requires is $O(nm \log m)$.
\end{proof}

\begin{algorithm}[h!]
\caption{$F_2$-Heavy Hitters}
\label{alg-hh}
\begin{algorithmic}[1]
\Require Stream $\bsig \in [n]^m$, threshold $\epsilon \in (0, 1)$
\State $g\leftarrow F_2(\bsig)$ ;\; $k \gets \lceil 2/\epsilon \rceil$\Comment{Compute global $F_2$ }
\State Partition $[n]$ into $k$ disjoint buckets $\mathcal{B} = \{B_1, \dots, B_k\}$.
Bucket $B_i$ contains the elements from $l_i=(i-1) \cdot (n / k) + 1$
to $r_i= i \cdot (n/k)$ and is represented by the pair of indices
$\{l_i, r_i\}$.
\State $H \gets \emptyset$
\While{$\mathcal{B} \neq \emptyset$}
    \State $\mathcal{B}_{next} \gets \emptyset$
    \State  $b_1, \ldots, b_{|\mathcal{B}|}\leftarrow F_2^\mathcal{B}(\bsig)$ \Comment{Compute the $F_2$ contribution of the elements in each bucket}
    \ForAll{$i\in \{1, \ldots, |\mathcal{B}|\}$}
        \If{$b_i \ge \epsilon \cdot g$}
                \If{$l_i = r_i$}
                     $H \gets H \cup \{l_i\}$
                \Else
               \;      $mid \gets \lfloor (l_i + r_i) / 2 \rfloor$
                   ;   $\mathcal{B}_{next} \gets \mathcal{B}_{next} \cup \{[l_i, mid], [mid + 1, r_i]\}$
            \EndIf
            \EndIf
            \EndFor
    \State $\mathcal{B} \gets \mathcal{B}_{next}$
\EndWhile
\State \Return $H$
\end{algorithmic}
\end{algorithm}
\textbf{Extension to $F_k$-heavy hitters.} 
Notice we can use the same idea in order to compute $\epsilon F_k-$Heavy Hitters using Algorithm \ref{alg:calculating moments}. In this case we would need $4 + 4\log \epsilon n$ passes and  $O(kn)$ 
catalytic registers over the ring $Z_q$ with $q = O(nm^k)$.

\section{Concluding Remarks}
In this paper, we introduced the \emph{catalytic streaming model}, which equips algorithms with a large auxiliary memory that must be restored to its initial state by the end of the computation. The central question of this model is identifying which streaming problems inherently benefit from this restricted memory. 

While recent results in catalytic communication~\cite{catalyticCommunication} imply that catalytic memory offers no advantage for one-pass algorithms, we demonstrated that it is remarkably powerful given multiple passes. Specifically, we showed how to exactly evaluate multivariate polynomials of stream frequencies using $O(1)$ regular memory of logarithmically many bits.

We  characterize a broad family of algorithms for which computing $F_2$ in two passes is impossible, and show that this barrier can be overcome by designing an algorithm outside this family. 
Since all our other algorithms satisfy these restriction (or extensions of them to more passes) we believe that this family may still be of interest.

We believe the catalytic streaming model opens a rich landscape for future work. Several compelling open problems remain:
\begin{enumerate}

    \item \textbf{Higher-degree pass optimality:} Our current lower bound technique is limited to two passes
    and puts restrictions on the algorithms that it applies to. Can one prove more general lower bounds that relate the number of required passes to the degree of 
the evaluated polynomial?
  %
    \item \textbf{Exact Heavy Hitters:} Can heavy hitters be computed exactly using a constant number of passes in the catalytic model?
    \item \textbf{Further catalytic separations:} Which other fundamental streaming problems admit catalytic algorithms that require substantially less regular memory than their standard streaming counterparts?
\end{enumerate}

\bibliographystyle{alpha}
\bibliography{refs} 

\appendix

\section{Powering Lemma} \label{app:powering}
We give here the powering algorithm of \autoref{Powering Lemma}. We denote $c_i \coloneqq (-1)^i {k\choose i}$.

\begin{algorithm}[H]
\caption{Powering Algorithm (Lemma 10) [Modified]} \label{alg:powering}
\begin{algorithmic}[1]
\Statex \textbf{Initial State:} $r =\tau, \; r_o = 0, \; \forall i \in [k] :  r_i = \tau_i$
\State \label{item:powering algo line 1} \textcolor{red}{$I_1$}: \textbf{for} $i = 1$ \textbf{to} $k$ \textbf{do} $r_o \gets r_o + c_i\cdot r_i \cdot r^i$
\Comment{ $r_o = \sum_{i=1}^k(-1)^{i}\binom{k}{i}\tau_i\cdot \tau^i$}

\State Execute $r\leftarrow r+f$
 \Comment{$r= \tau + f$}

\State \label{item:powering algo line 3} \textcolor{red}{$I_2$}: \textbf{for} $i = 1$ \textbf{to} $k$ \textbf{do} $r_i \gets r_i - r^{k-i}$
\Comment{ $\forall i\in [k], \; r_i= \tau_i - (\tau+f)^{k-i}$}
\State \label{item:powering algo line 4} \textcolor{red}{$I_2$}: $r_o \gets r_o + r^k$
\Comment{ $r_o = (\tau + f)^k + \sum_{i=1}^{k} (-1)^i \binom{k}{i}\tau_i \tau^i$}

\State Execute $r\leftarrow r - f$
\Comment{$r= \tau $}

\State \label{item:powering algo line 6} \textcolor{red}{$I_3$}: \textbf{for} $i = 1$ \textbf{to} $k$ \textbf{do} $r_o \gets r_o - c_i\cdot r_i\cdot r^i$

\Statex \Comment{ $r_o = (\tau + f)^k + \sum_{i=1}^{k} (-1)^i \binom{k}{i}\tau_i \tau^i -  \sum_{i=1}^{k} (-1)^i \binom{k}{i}(\tau_i - (\tau + f)^{k-i})\tau^i$}
\Statex \Comment{ $=    (\tau + f)^k + \sum_{i=1}^{k} (-1)^i \binom{k}{i}\tau^i (\tau + f)^{k-i}$ }
\end{algorithmic}
\label{PoweringAlg}
\end{algorithm}

After the execution of line~\ref{item:powering algo line 1}, the output register satisfies
\[
r_o=\sum_{i=1}^k (-1)^i \binom{k}{i}\tau_i \tau^i.
\]

After the execution of line~\ref{item:powering algo line 3}, we have
\[
\forall i\in [k], \qquad r_i=\tau_i-(\tau+f)^{k-i}.
\]

After the execution of line~\ref{item:powering algo line 4}, the output register satisfies
\[
r_o=(\tau+f)^k+\sum_{i=1}^k (-1)^i \binom{k}{i}\tau_i \tau^i.
\]

Hence, after the execution of line~\ref{item:powering algo line 6}, we obtain
\begin{align*}
r_o
&= (\tau+f)^k+\sum_{i=1}^k (-1)^i \binom{k}{i}\tau_i \tau^i
   -\sum_{i=1}^k (-1)^i \binom{k}{i}\bigl(\tau_i-(\tau+f)^{k-i}\bigr)\tau^i \\
&= (\tau+f)^k+\sum_{i=1}^k (-1)^i \binom{k}{i}\tau^i(\tau+f)^{k-i} \\
&= f^k,
\end{align*}
where the last equality follows from the binomial theorem.

Therefore, at the end of the computation, the registers contain
\[
r=\tau, \qquad r_o=f^k, \quad \text{and for every $i\in [k]$,} \quad r_i=\tau_i-(\tau+f)^{k-i}.
\]  

\section{\textsc{Set-Disjointness}}
\label{app:dis}

\begin{proposition} \label{prop:cat disj}
    Suppose there exists a catalytic streaming algorithm $A$ that computes $F_2$ in $p$ passes using $s$ bits of regular memory. Then there exists a catalytic communication protocol for $\disj_n$ with $2p$ rounds that uses $s + \log (n) + c_0$ bits of regular memory, where $c_0$ is some absolute constant.
\end{proposition}

\begin{proof}
    Let $\x$ be Alice's input and $\y$ be Bob's input. To compute $\disj_n(\x,\y)$, Alice and Bob simulate the streaming algorithm $A$ on an appropriate stream, while using the shared memory exactly as $A$ would.

    First, Alice inserts into the stream every item $i \in [n]$ such that $x_i = 1$. She then passes the inner state of $A$ plus it's memory configuration to Bob.
    
    Bob then inserts in a similar way $\y$ into the stream. Moreover, he calculates $\lVert\y\rVert_0$\footnote{The $L_0$ norm of $\y$, i.e. the number of $1$ bits in $\y$} and passes both $A$ and $\lVert\y\rVert_0$ to Alice. 

    They continue in this way, alternately simulating the execution of $A$ and passing it's inner state together with it's current configuration, until the execution of $A$ is complete.
    
    

    Let $a$ denote the output of $A$ on the simulated stream. Alice accepts if and only if
    \[
    a = \|\x\|_0 + \|\y\|_0.
    \]
    Note that $\|\x\|_0 + \|\y\|_0$ is exactly the length of the simulated stream.

    \begin{claim}
    Alice accepts if and only if $\disj_n(\x,\y)=1$.
    \end{claim}

\begin{proof}
Suppose first that $\disj_n(\x,\y)=1$. Then the supports of $\x$ and $\y$ are disjoint, so each item $i \in [n]$ appears in the stream at most once. Hence, for every $i \in [n]$, we have \( f_i \in \{0,1\}, \) and therefore \(f_i^2 = f_i.\)
It follows that
\[
F_2=\sum_{i=1}^n f_i^2 = \sum_{i=1}^n f_i = \|\x\|_0+\|\y\|_0.
\]
Since $A$ computes $F_2$, we have $a=\|\x\|_0+\|\y\|_0$, and so Alice accepts.

Conversely, suppose that $\disj_n(\x,\y)=0$. Then there exists some $i \in [n]$ such that $x_i=y_i=1$, and hence $f_i=2$. More generally, there is at least one index $i$ for which $f_i>1$. Since for every $t \in \mathbb{Z}_{\ge 0}$, \(t^2 \ge t,\)with strict inequality whenever $t>1$, we obtain
\[
F_2=\sum_{i=1}^n f_i^2 > \sum_{i=1}^n f_i = \|\x\|_0+\|\y\|_0.
\]
Thus $a > \|\x\|_0+\|\y\|_0$, and Alice rejects.
\end{proof}
    



    The claim about the number of rounds is trivial. Each time the simulation is passed from one player to the other, the players need to communicate the current configuration of $A$, which contributes $s + c_0$ bits of regular memory, $s$ bits for the memory size of $A$ and $c_0$ for description of it's inner state. In addition, Bob communicates the value $\|\y\|_0$ to Alice, which requires at most $\log n$ bits of regular memory.
\end{proof}

\begin{remark}
    Observe that the last reduction works for algorithm $A$ which computes any moment other than $F_1$, not just $F_2$.
\end{remark}

\section{A Three-Round Catalytic Communication Protocol for \textsc{Set-Disjointness}}
\label{app:inner product} 

In \cite[Proposition~5]{catalyticCommunication}, it is shown that the inner product of two vectors over $\text{GF}_2$ can be computed by a three-round catalytic communication protocol using only $1$ bit of regular memory and $n$ bits of catalytic memory. We observe that essentially the same construction, with a minor modification, yields a protocol for $\disj_n$. For completeness, we include the proof here, following their argument with the necessary changes.

\begin{proposition}
There exists a catalytic communication protocol for $\disj_n$ that uses $\log n + 1$ bits of regular memory, $n(\log n + 1)$ bits of catalytic memory, and requires $3$ rounds.
\end{proposition}

\begin{proof}
Let $\x \in \{0,1\}^n$ be Alice's input and let $\y \in \{0,1\}^n$ be Bob's input.

Partition the catalytic tape into $n$ registers
\[
r_1,r_2,\dots,r_n,
\]
each of size $\log (n) + 1$ bits. For each $i \in [n]$, let $\tau_i$ denote the initial contents of register $r_i$. The regular memory consists of $\log (n) + 1$ bits, so throughout the protocol all arithmetic operations are performed modulo $2n$. Alice and Bob exchange the clean abd catalytic memory in three rounds  as follows.

\begin{enumerate}
    \item \textbf{Alice} increments $r_i$ for every $i\in [n]$ such that
    $x_i =1$. 

    \item 
    \textbf{Bob} adds $r_i$ to regular memory for every $i\in [n]$ such that
    $y_i=1$. 
 
Following this round, the regular memory contains
    \begin{equation}\label{eq:disj protocol second pass memory eq}
        \sum_{i:\, y_i=1} r_i \pmod{2n}
        =
        \sum_{i:\, y_i=1} \bigl(\tau_i + \mathbf{1}[x_i=1]\bigr)
        \pmod{2n},
    \end{equation}
    where $\mathbf{1}[x_i=1]$ is the indicator of the event $x_i=1$.
    \item 
    \textbf{Alice} decrements $r_i$  
     for every $i\in [n]$ such that $x_i=1$.

Following this step, every catalytic register is restored to its initial value. .
\end{enumerate}

Finally, once Bob gets the second message from Alice, he
 Bob subtracts $r_i$ from the value in the regular memory for every $i$ such that $y_i = 1$. After doing so, using Equation~\eqref{eq:disj protocol second pass memory eq}, the value Bob obtains is
\begin{equation} \label{eq:disj protocol reg memory output}
    \sum_{i \, : \, y_i=1} \mathbf{1}[x_i=1] \pmod{2n} \, .
\end{equation}
Bob outputs $1$ if this value is $0$, and outputs $0$ otherwise.
    
The space bounds are immediate from the construction, the catalytic memory contains $n$ registers of $\log n+1$ bits each, for a total of $n(\log n+1)$ bits, while the regular memory uses exactly $\log n+1$ bits.

    To prove correctness, observe that
\[
\sum_{i:\, y_i=1} \mathbf{1}[x_i=1]
=
|\{i \in [n] : x_i=y_i=1\}|,
\]
namely, the number of coordinates in which both $\x$ and $\y$ contain a $1$. Therefore,
\[
\disj_n(\x,\y)=1
\quad\Longleftrightarrow\quad
\sum_{i:\, y_i=1} \mathbf{1}[x_i=1]=0.
\]
Moreover, this sum is always at most $n$, and hence is strictly smaller than $2n$. Thus its value modulo $2n$ is zero if and only if the sum itself is zero.
\end{proof}

\section{Why the Discrete Derivative Argument Breaks in Three Passes}
\label{proof break}
The impossibility result strictly relies on the two-pass constraint (one to hold some information about the input and one to clean the catalytic tape). We claim that if the algorithm is permitted a third pass, the information-theoretic contradiction vanishes.
In a three-pass model, we have three bijection matrices $U_{i,1}$, $U_{i,2}$, $U_{i,3}$, one for each pass such that\ for every $f_i$, and every $\tau_i$, applying $(U_{i,3})^{f_i}\Big((U_{i,2})^{f_i}\big((U_{i,1})^{f_i}(\tau_i)\big)\Big) = \tau_i$. 
If we apply the same proof technique then the derivative we get in this case will depend on $y_1 = (U_{i,1})^{f_i}(\tau_1)$ and $y_2 = (U_{i,2})^{f_i}\big((U_{i,1})^{f_i}(\tau_i)\big) $. So, we will have some deterministic function $H_i(\cdot,\cdot)$ that satisfies $H_i(y_1,y_2) = 2f_i + 1$. 
Crucially, while $y_1$ and $y_2$ are individually uniformly distributed and independent of $f_i$,  they jointly can determine $f_i$. 
In particular they indeed determine $f_i$ if there is a unique $k$ for which 
$y_2 = (U_{i,2})^k(y_1)$.

\end{document}